\documentclass{article}
\usepackage{fullpage}
\usepackage{amsthm,amssymb,amsmath}  
\usepackage{xspace,enumerate}
\usepackage[utf8]{inputenc}
\usepackage{authblk}
\usepackage{url}
\usepackage{algorithm,algorithmicx,algpseudocode}
\usepackage{cite}
\usepackage{tikz}
\usetikzlibrary{decorations.pathreplacing,patterns,shapes}

  \theoremstyle{plain}
  \newtheorem{theorem}{Theorem}
  \newtheorem{lemma}{Lemma}  
  \newtheorem{corollary}[theorem]{Corollary}  
  \newtheorem{fact}{Fact}
  \newtheorem{observation}{Observation}
  \theoremstyle{definition}
  \newtheorem{definition}{Definition}

  \newtheorem{remark}[definition]{Remark}

\title{Streaming dictionary matching with mismatches\thanks{This is a full and extended version of the conference paper \cite{DBLP:conf/cpm/GawrychowskiS19}.}
\thanks{P. Gawrychowski was partially supported by the Bekker programme of the Polish National Agency for Academic Exchange (PPN/BEK/2020/1/00444) and the grant ANR-20-CE48-0001 from the French National Research Agency (ANR). T. Starikovskaya was partially supported by the grant ANR-20-CE48-0001 from the French National Research Agency (ANR).}}
\author[1]{Pawe\l{} Gawrychowski}
\author[2]{Tatiana Starikovskaya}

\affil[1]{University of Wrocław, 50-137 Wrocław, Poland\\\texttt{gawry@cs.uni.wroc.pl}}
\affil[2]{DI/ENS, PSL Research University, Paris, France\\
    \texttt{tat.starikovskaya@gmail.com}
}

\date{\vspace{-5ex}}

{\itshape}

\newcommand{\polylog}[1]{\mathop{\mathrm{polylog} {\;#1}}}
\newcommand{\Oh}{\mathcal{O}}
\newcommand{\Ooh}{\tilde{\Oh}}
\newcommand{\prefixsearch}{\mathsf{PrefixSearch}}
\newcommand{\occ}{\mathrm{output}}
\newcommand{\Ham}{\mathsf{H}}
\makeatletter
\newcommand*{\rom}[1]{\expandafter\@slowromancap\romannumeral #1@}
\makeatother

\usepackage{doi}
\begin{document}
\maketitle
\begin{abstract}
In the $k$-mismatch problem we are given a pattern of length $n$ and a text and must find all locations where the Hamming distance between the pattern and the text is at most $k$. A series of recent breakthroughs have resulted in an ultra-efficient streaming algorithm for this problem that requires only $\Oh(k \log \frac{n}{k})$ space and $\Oh(\log \frac{n}{k} (\sqrt{k \log k} + \log^3 n))$ time per letter [Clifford, Kociumaka, Porat, SODA 2019]. In this work, we consider a strictly harder problem called dictionary matching with $k$ mismatches. In this problem, we are given a dictionary of $d$ patterns, where the length of each pattern is at most $n$, and must find all substrings of the text that are within Hamming distance $k$ from one of the patterns. We develop a streaming algorithm for this problem with $\Oh(k d \log^k d \polylog{n})$ space and $\Oh(k \log^{k} d \polylog n + |\occ|)$ time per position of the text. The algorithm is randomised and outputs correct answers with high probability. On the lower bound side, we show that any streaming algorithm for dictionary matching with $k$ mismatches requires $\Omega(k d)$ bits of space. 
\end{abstract}

\section{Introduction}
\label{sec:introduction}
In the fundamental dictionary matching problem, we are given a dictionary of patterns and a text, and must find all the substrings of the text equal to one of the patterns (such substrings are called \emph{occurrences} of the patterns). 
The classical algorithm for dictionary matching is the one by Aho and Corasick~\cite{AC75}.
For a dictionary of $d$ patterns of length at most $n$, their algorithm uses $\Omega(n d)$ space and $\Oh(1+|\occ|)$ time per letter, where $\occ$ is the set of  occurrences of the patterns that end at this position. Apart from the Aho--Corasick algorithm, other word-RAM algorithms for exact dictionary matching include~\cite{BAP10,B12,BR13,CW79,CCGLPR99,FGGK15,DBLP:conf/wads/GolanKKP19,HKSTV10,KPR16,kosolobov_et_al:LIPIcs:2019:10484,WU92}. 

However, in many applications one is interested in substrings of the text that are close to but not necessarily equal to the patterns. This task can be naturally formalised as follows: given a dictionary of $d$ patterns of length at most $n$, and a text, find all substrings of the text within  distance $k$ from one of the patterns, where the distance is either the Hamming or the edit distance. In this work, we focus on the Hamming distance, and refer to this problem as \emph{dictionary matching with $k$ mismatches}. 

We give a brief survey of existing solutions in the word-RAM model, ignoring the special case of $k = 1$ that relies on very different techniques. The case of $d = 1$ was considered in~\cite{DBLP:journals/jal/AmirLP04,k-mismatch,gawrychowski_et_al:LIPIcs:2018:9066,LANDAU1986239}. The latest algorithm~\cite{gawrychowski_et_al:LIPIcs:2018:9066} uses $\Oh(n)$ space and $\Oh(\log^2 n+k \sqrt{\log (n) / n})$ amortised time per letter for constant-size alphabet. 
These algorithms can be generalised to~$d > 1$ patterns by running $d$ instances of the algorithm in parallel. One can also reduce the problem to dictionary look-up with $k$ mismatches or text indexing with $k$ mismatches. In the former problem, the task is to preprocess the dictionary of patterns into a data structure to support the following queries fast: given a string $Q$, find all patterns in the dictionary within Hamming distance $k$ from~$Q$. In the latter, the task is to preprocess the text so that given a pattern to be able to report all substrings of the text within Hamming distance $k$ from the pattern efficiently. These problems were considered in~\cite{kerratatree,EPIFANIO2007152,DBLP:conf/mfcs/GawrychowskiLS18,HHLS:2006,LSW:2008,TSUR2010339}. 
However, all of the above algorithms must at least store the dictionary in full, which in the worst case requires $\Omega(nd)$ bits of space. 

In this work, we focus on the streaming model of computation that was designed to overcome this restriction and allows developing particularly efficient algorithms. In the streaming model, we assume that the text arrives as a stream, one letter at a time. 
The space complexity of an algorithm is defined to be all the space used, including the space we need for storing the information about the pattern(s) and the text. The time complexity of an algorithm is defined to be the time we spend to process one letter of the text. The streaming model of computation aims for algorithms that use as little space and time as possible. All streaming algorithms we discuss in this paper are randomised and output correct answers with high probability\footnote{With high probability means with probability at least $1-1/n^c$ for any predefined constant $c>1$.}. 

Throughout the paper, we assume that the length of the text is $\Oh(n)$. If the text is longer, we can partition it into overlapping blocks of length $\Oh(n)$ and process each block independently.  

The first sublinear-space streaming algorithm for the dictionary matching problem with $d = 1$ was suggested by Porat and Porat~\cite{Porat:09}. For a pattern of length $n$, their algorithm uses $\Oh(\log n)$ space and $\Oh(\log n)$ time per letter. Later, Breslauer and Galil gave a $\Oh(\log n)$-space and $\Oh(1)$-time algorithm~\cite{Breslauer:11}. For arbitrary~$d$, Clifford et al.~\cite{streamdictionary} showed a streaming algorithm that uses $\Oh(d \log n)$ space and $\Oh(\log\log(n+d)+|\occ|)$ time per letter. Golan and Porat~\cite{GP17} showed an improved algorithm that uses the same amount of space and $\Oh(1+|\occ|)$ time per letter for constant-size alphabets.

Dictionary matching with $k$ mismatches has been mainly studied for~$d = 1$. The first algorithm was shown by Porat and Porat~\cite{Porat:09} by reduction to exact dictionary matching. The algorithm uses $\Oh(k^3 \log^7 n / \log \log n)$ space and $\Oh(k^2 \log^5 n / \log \log n)$ time. The complexity has been subsequently improved in~\cite{k-mismatch,Clifford:17,DBLP:conf/icalp/GolanKP18}. The current best algorithm uses only $\Oh(k \log \frac{n}{k})$ space and $\Oh(\log \frac{n}{k} (\sqrt{k \log k} + \log^3 n))$ time per letter~\cite{Clifford:17}. Golan et al.~\cite{DBLP:conf/cpm/GolanKKP20} studied space-time trade-offs for this problem. 
For $d > 1$, one can obtain the following result by a repeated application of the algorithm for $d = 1$~\cite{Clifford:17}:

\begin{corollary}\label{cor:streaming-k-mismatch}
For any $k \ge 1$, there is a randomised streaming algorithm for dictionary matching with $k$ mismatches that uses $\Ooh(d k)$ space and $\Ooh(d \sqrt{k})$ time per letter\footnote{Hereafter, $\Ooh$ hides a multiplicative factor polynomial in $\log n$.}. The algorithm has two-sided error and outputs correct answers with high probability.
\end{corollary}

\subsection{Our results}
In this work, we consider the problem of streaming dictionary matching with $k$ mismatches for arbitrary $d > 1$. 

As it can be seen, the time complexity of Corollary~\ref{cor:streaming-k-mismatch} depends on $d$ linearly, which is prohibitive for applications where the stream letters arrive at a high speed and the size of the dictionary is large, up to several thousands of patterns, as we must be able to process each letter before the next one arrives to benefit from the space advantages of streaming algorithms. In this work, we show an algorithm that uses $\Ooh(k d \log^k d)$ space and $\Ooh(k \log^{k} d + |\occ|)$ time per letter, assuming polynomial-size alphabet (Theorem~\ref{th:main-deamortised}). Our algorithm makes use of a new randomised variant of the $k$-errata tree (Section~\ref{sec:kerratatree}), a famous data structure of Cole, Gottlieb, and Lewenstein for dictionary matching with $k$ mismatches~\cite{kerratatree}. This variant of the $k$-errata tree allows to improve both the query time and the space requirements and can be considered as a generalisation of the $z$-fast tries~\cite{zfasttrie1,zfasttrie}, that have proved to be useful in many streaming applications. 

We also show that any streaming algorithm for dictionary matching with $k$ mismatches requires $\Omega(k d)$ bits of space (Lemma~\ref{lm:space_lower_bound}). This lower bound implies that for constant values of $k$ our algorithm is optimal up to polylogarithmic factors.

\section{Preliminaries}
\label{sec:prelim}
In this section, we give the definitions of strings, tries, and two hash functions that we use throughout the paper: Karp--Rabin fingerprints~\cite{KarpRabin} and sketches for the Hamming distance~\cite{Clifford:17}. 

\subsection{Strings and tries}
We assume an integer alphabet $\{1, 2, \dots, \sigma\}$ of size $\sigma = n^{\Oh(1)}$. A \emph{string} is a finite sequence of letters of the alphabet. For a string $S = S[1] S[2] \ldots S[m]$ we denote its length $m$ by $|S|$ and its substring $S[i] S[i+1] \ldots S[j]$, $1\le i < j \le m$, by $S[i,j]$. If $i = 1$, the substring $S[1,j]$ is referred to as a \emph{prefix} of~$S$. If $j = m$, $S[i,m]$ is called a \emph{suffix} of $S$. We say that a substring $S[i,j]$ is an \emph{occurrence} of a string $X$ in $S$ if  $S[i,j] = X$ and a \emph{$k$-mismatch occurrence} of $X$ in $S$ if the Hamming distance between $S[i,j]$ and $X$ is at most $k$. (Recall that the Hamming distance between two strings $X, Y$ of equal lengths is defined as the number of mismatches between them, in other words, as the number of positions where $X$ and~$Y$ differ.) The reverse of a string $S = S[1] S[2] \ldots S[m]$, denoted by $S^R$, is defined as $S[m] S[m-1] \ldots S[1]$. We denote the concatenation of strings $X,Y$ by $X \circ Y$. We use notation $X^m$ for the concatenation of $m$ copies of a string $X$.

A \emph{trie} is a basic data structure that can be used to store a set of strings. A trie is a tree satisfying the following properties:

\begin{enumerate}
\item Each edge is labelled by a letter of the alphabet;
\item Each two edges outgoing from the same node are labelled by different letters;
\item For each string $S$ in the set there is a node $v_S$ of the trie such that the concatenation of the labels of the edges from the root to $v_S$ is equal to $S$, and the concatenation of the labels in any path starting at the root of the trie is equal to a prefix of some string in the set. 
\end{enumerate}

The number of nodes in a trie is in the worst case linear in the total length of the strings. To improve the space requirements, we compactify the trie: namely, for every string~$S$ we mark the node $v_S$, and then replace each maximal path of unmarked nodes of degree one with an edge labelled by the concatenation of the letters on the edges in the path. The result is called a \emph{compact trie} (see Fig.~\ref{fig:trie} for an example). We call the nodes of the compact trie \emph{explicit nodes}, and the nodes that were deleted during the construction \emph{implicit nodes}. The number of the explicit nodes of the compact trie is linear in the number of strings it stores. 

\begin{figure}
\begin{center}
\begin{minipage}{0.4\textwidth}
\begin{tikzpicture}[every node/.style={draw=black,diamond,inner sep=0pt,minimum size=5pt}]
\node[fill=white] (root) at (0,0) {};
\node[fill=black] (v1) at (-1,-1) {};
\node[fill=white] (v2) at (1,-1) {};
\node[fill=black] (v3) at (-2,-2) {};
\node[fill=black] (v4) at (0,-2) {};
\node[fill=black] (v5) at (2,-2) {};

\node[draw=black,fill=black,circle,inner sep=0pt,minimum size=2pt] (u1) at (-0.5,-0.5) {};
\node[draw=black,fill=black,circle,inner sep=0pt,minimum size=2pt] (u2) at (-1.5,-1.5) {};
\node[draw=black,fill=black,circle,inner sep=0pt,minimum size=2pt] (u3) at (0.5,-1.5) {};
\node[draw=black,fill=black,circle,inner sep=0pt,minimum size=2pt] (u4) at (1.5,-1.5) {};

\draw (root) -- (u1) node[pos=0.5,rectangle,sloped,fill=white,draw=white] {a};
\draw (u1) -- (v1) node[pos=0.5,rectangle,sloped,fill=white,draw=white] {a};
\draw (root) -- (v2) node[pos=0.5,rectangle,sloped,fill=white,draw=white] {b};
\draw (v1) -- (u2)  node[pos=0.5,rectangle,sloped,fill=white,draw=white] {b};
\draw (u2) -- (v3)  node[pos=0.5,rectangle,sloped,fill=white,draw=white] {c};
\draw (v2) -- (u3)  node[pos=0.5,rectangle,sloped,fill=white,draw=white] {a};
\draw (u3) -- (v4)  node[pos=0.5,rectangle,sloped,fill=white,draw=white] {c};
\draw (v2) -- (u4)  node[pos=0.5,rectangle,sloped,fill=white,draw=white] {b};
\draw (u4) -- (v5)  node[pos=0.5,rectangle,sloped,fill=white,draw=white] {b};
\end{tikzpicture}
\end{minipage}
\begin{minipage}{0.4\textwidth}
\begin{tikzpicture}[every node/.style={draw=black,diamond,inner sep=0pt,minimum size=5pt}]
\node[fill=white] (root) at (0,0) {};
\node[fill=black] (v1) at (-1,-1) {};
\node[fill=white] (v2) at (1,-1) {};
\node[fill=black] (v3) at (-2,-2) {};
\node[fill=black] (v4) at (0,-2) {};
\node[fill=black] (v5) at (2,-2) {};

\draw (root) -- (v1) node[pos=0.5,rectangle,sloped,fill=white,draw=white] {aa};
\draw (root) -- (v2) node[pos=0.5,rectangle,sloped,fill=white,draw=white] {b};
\draw (v1) -- (v3)  node[pos=0.5,rectangle,sloped,fill=white,draw=white] {cb};
\draw (v2) -- (v4)  node[pos=0.5,rectangle,sloped,fill=white,draw=white] {ca};
\draw (v2) -- (v5)  node[pos=0.5,rectangle,sloped,fill=white,draw=white] {bb};
\end{tikzpicture}
\end{minipage}
\end{center}
\caption{The trie (left) and the compact trie (right) for a dictionary $\{aa, aabc, bac, bbb\}$. The nodes of the trie that we delete are shown by circles, marked nodes (nodes labelled by the dictionary patterns) are black.}
\label{fig:trie}
\end{figure}

\subsection{Fingerprints and sketches}
Let us first give the definition of Karp--Rabin fingerprints that can be used to decide whether two strings are equal.

\begin{definition}[Karp--Rabin fingerprints~\cite{KarpRabin}]\label{def:kr}
For a fixed prime $p$ and $r \in [0,p-1]$ chosen u.a.r., the Karp–-Rabin fingerprint of a string $X = X[1] X[2] \dots X[m]$ is defined as a quadruple $\Phi(X) = (\varphi(X), \varphi(X^R), r^{m} \bmod p, r^{-m} \bmod p)$, where $\varphi(X) = \sum_{i = 1}^m X[i] \cdot r^{m-i} \bmod p$ and $\varphi(X^R) = \sum_{i = 1}^m X[i] \cdot r^{i-1} \bmod p$.
\end{definition}

Note that this definition slightly differs from the standard one and in particular  given the Karp--Rabin fingerprint of $X$ we can compute the Karp--Rabin fingerprint of $X^R$ in $\Oh(1)$ time and space (by simply reversing the order of $\varphi(X)$ and $\varphi(X^R)$).

\begin{fact}[{\cite{Porat:09}}]\label{fct:KR}
For $r \in [0,p-1]$ chosen u.a.r., the probability of two distinct strings of equal lengths $m \le n$ over the integer alphabet $[0, p-1]$ to have equal Karp--Rabin fingerprints is at most $n/p$.
\end{fact}

Below we assume that $p$ is chosen sufficiently large to guarantee that the collision probability is inverse-polynomial in $n$. 

\begin{fact}\label{fct:KR_concat}
We can construct one of the fingerprints $\Phi(X)$, $\Phi(Y)$, or $\Phi(X \circ Y)$ given the other two in $\Oh(1)$ time and space. 
\end{fact}

We now remind the definition of $k$-mismatch sketches that can be used to decide whether two strings are at Hamming distance at most $k$.

\begin{definition}[$k$-mismatch sketch~\cite{Clifford:17}]
For a fixed prime $p$ and $r \in [0,p-1]$ chosen u.a.r., the $k$-mismatch sketch $sk_{k}(S)$ of a string $S$ of length $m$ is defined as a tuple $(\phi_0(S), \ldots, \phi_{2k}(S)$, $\phi'_0(S), \ldots, \phi'_{k}(S), \Phi(S))$, where $\phi_j(S) = \sum_{i = 1}^{m} S[i] \cdot i^{j} \bmod p$ and $\phi'_j(S) = \sum_{i = 1}^{m} S[i]^2 \cdot i^{j} \bmod p$ for $j \ge 0$.
\end{definition}

\begin{lemma}[{\cite[Proposition 3.1]{Clifford:17}}]\label{lm:compute_dist}
Let $X,Y$ be two strings of length $m \le n$. Given the sketches $sk_k(X)$ and $sk_k(Y)$, there is a randomised $\Ooh(k)$-time and $\Oh(k)$-space algorithm that reports the set $\{(p, X[p], Y[p]) : X[p] \neq Y[p]\}$ if the Hamming distance between $X$ and $Y$ is at most $k$, and otherwise it simply outputs a message ``The distance is larger than~$k$.''. The algorithm has two-sided error and outputs correct answers with high probability.
\end{lemma}

\begin{corollary}[{of~\cite[Proposition 3.1]{Clifford:17}}]\label{cor:sketch_concat}
We can construct one of the sketches $sk_k(X)$, $sk_k(Y)$, or $sk_k(X \circ Y)$ given the other two in $\Ooh(k)$ time using $\Oh(k)$ space, provided that all concerned strings are over the alphabet $[0,p-1]$ and are of length at most~$n$. Furthermore, we can compute $sk_k(X^m)$ in $\Ooh(k)$ time and $\Oh(k)$ space as well under the same assumption. 
\end{corollary}

\section{The randomised $k$-errata tree}\label{sec:kerratatree}
Cole, Gottlieb, and Lewenstein~\cite{kerratatree} introduced a data structure called the $k$-errata tree. The $k$-errata tree is a data structure that supports dictionary look-up with $k$ mismatches queries: Given a query string $Q$, find all patterns in the dictionary that are at Hamming distance at most~$k$ from $Q$ or one of its prefixes\footnote{The query algorithm of~\cite{kerratatree} returns only those patterns that are within Hamming distance $k$ from $Q$ itself, but considering prefixes as well does not change the query time and is more suitable for our purposes. We explain the necessary modifications in Section~\ref{sec:reminderkerratatree}.}. In this section, we introduce a new data structure  which we call the randomised $k$-errata tree. This data structure uses less space than the $k$-errata tree, has better query time (under a certain assumption, see below), and outputs the desired patterns correctly with high probability. 

\subsection{Reminder: the $k$-errata tree}\label{sec:reminderkerratatree}
For completeness, we first remind the definition of the $k$-errata tree and give an outline of the query algorithm. The $k$-errata tree is based on the heavy-path decomposition of a tree:

\begin{definition}
The heavy path of a tree $\mathcal{T}$ is the path that starts at the root of $\mathcal{T}$ and at each node $v$ on the path descends to the child with the largest number of leaves in its subtree (\emph{heavy} child), with ties broken arbitrarily. The heavy path decomposition is defined recursively, namely, it is defined to be a union of the heavy path of $\mathcal{T}$ and the heavy path decompositions of the off-path subtrees of the heavy path. 
\end{definition}

A well-known property of the heavy-path decomposition is that any root-to-leaf path crosses $\Oh(\log |\mathcal{T}|)$ heavy paths, where $|\mathcal{T}|$ is the number of nodes of $\mathcal{T}$. This property is essential for the analysis of the space complexity and the query time of the $k$-errata tree.

\paragraph{Data structure.}
Consider a dictionary $\mathcal{D}$ of $d$ patterns of maximal length $m$. We start with the compact trie~$\mathcal{T}$ for the dictionary $\mathcal{D}$, and decompose it into heavy paths. 

During the recursive step, we construct a number of new compact tries. For each heavy path $H$, and for each node $u \in H$ consider the off-path trees hanging from $u$. First, we create a \emph{vertical substitution trie} for $u$. Let $a$ be the first letter on the edge $(u,v) \in H$. Consider an off-path tree hanging from $u$, and let $b \neq a$ be the first letter on the edge from $u$ to this tree. For each pattern in this off-path tree, we replace $b$ with~$a$. We consider a set of patterns obtained by such  substitution for all off-path trees hanging from~$u$ and build a new compact trie for this set, which we call the vertical substitution trie. 

Next, we create \emph{horizontal substitution tries} for the node $u$. We create a separate horizontal substitution trie for each off-path tree hanging from~$u$. To do so, we take the patterns in it and cut off the first letters up to and including the first letter on the edge from $u$ to this tree, and then build a compact trie on the resulting set of patterns. 

We finally group the substitution tries. In more detail, for each heavy path we consider its vertical substitution tries and build a weight-balanced tree, where the leaves of the weight-balanced tree are the vertical substitution tries, in the top-down order, and for each node of the tree, we create a new trie by merging the tries below~it. For each of these group vertical substitution tries we build the $(k-1)$-errata tree. We group the horizontal substitution tries in a similar way, namely, we consider each node $u$ and build a weight-balanced tree on the horizontal substitution tries that we created for the node $u$. 

From the construction, it follows that the $k$-errata tree is a set of compact tries, and each string $S$ in the tries originates from a pattern in the dictionary $\mathcal{D}$. We mark the end of the path labelled by $S$ by the id of the pattern it originates from. Moreover, one can show the following upper bound on the size of the $k$-errata tree that follows from the property of the heavy-path decomposition mentioned above:

\begin{lemma}[{\cite[Lemma 10]{kerratatree}}]\label{lm:kerratatreesize}
The size of the $k$-errata tree is $\Oh(d \log^k d)$.
\end{lemma}

\paragraph{Queries.} For simplicity, we first explain how to find the patterns in $\mathcal{D}$ that are within Hamming distance $k$ from $Q$ itself, and then explain how to modify the query algorithm to retrieve the patterns that are within Hamming distance $k$ from $Q$ or one of its prefixes.

The query algorithm is recursive, and is based on a procedure called $\prefixsearch$. This procedure takes three arguments: a compact trie, a starting node $u$ (explicit or implicit) in this trie, and a query string $Q'$, and outputs a pointer to the end of the longest path starting at $u$ and labelled by a prefix of $Q'$. 

For the purposes of recursion, we introduce a mismatch credit~--- the number of mismatches that we are still allowed to make. The algorithm starts with the mismatch credit $\mu = k$ and runs a $\prefixsearch$ in the trie $\mathcal{T}$ for the query string $Q$, starting from the root. If $\mu = 0$ and the path is labelled by $Q$, the algorithm returns the ids of the patterns in $\mathcal{D}$ that are associated with the end of the path. Otherwise, consider the heavy paths $H_1, H_2, \dots, H_j$ traversed by the $\prefixsearch$. Let $u_i$ be the last node of the heavy path $H_{i}$, $1 \le i \le j$,
visited by $\prefixsearch$. Note that for $i < j$, $u_i$ is necessarily an explicit node of $\mathcal{T}$, and for $i = j$ it can be an implicit node. Divide all the patterns in $\mathcal{D}$ into four groups: (\rom{1}) \label{tp:i} Patterns hanging off a node $u$ in a heavy path $H_i$, where~$u$ is located above $u_i$, $1 \le i \le j$; (\rom{2}) \label{tp:ii} Patterns in the subtrees of $u_i$'s children not in the heavy path $H_{i+1}$, for $1 \le i < j$; (\rom{3}) \label{tp:iii} Patterns in the subtree of the node in $H_j$ that is just below $u_j$; (\rom{4}) \label{tp:iv} If $u_j$ is a node, then patterns in the subtrees of $u_j$'s children not in the heavy path~$H_{j}$.

The algorithm processes each of the pattern groups independently. Consider a pattern $P$ in group~\rom{1}. Suppose that it hangs from a node $u \in H_i$, where $u$ is above~$u_i$, and let $\ell$ be the length of the label of $u$. We have that $Q$ and $P$ have a mismatch at the position~$\ell + 1$. When creating the vertical substitution trie for $u$, we removed this mismatch. Consider the weight-balanced tree for $H_i$ and the minimal set of nodes containing the vertical substitution tries for the nodes in $H_i$ above $u$. To finish the recursive step, we call the algorithm with the mismatch credit $\mu - 1$ for the $(k-1)$-errata trees that we built for these nodes. The patterns of groups~\rom{2} and~\rom{4} are processed in a similar way but using the $(k-1)$-errata trees for the horizontal substitution trees. Finally, to process the patterns of group~\rom{3}, we call the algorithm with mismatch credit~$\mu-1$ starting from the node that follows $u_j$ in~$H_j$.

\begin{lemma}[{\cite[Lemma 11]{kerratatree}}]\label{lm:kerratatreetime}
The query algorithm makes $\Oh(\log^k d)$ $\prefixsearch$ calls.
\end{lemma}

Cole, Gottlieb, and Lewenstein implemented $\prefixsearch$ deterministically:

\begin{lemma}[{\cite{kerratatree}}]
Using $\Oh(nd + d \log^k d)$ extra space, the $\Oh(\log^k d)$ $\prefixsearch$ calls can be answered in $\Oh(n + \log^k d \log \log n)$ time. 
\end{lemma}

We will use this solution for the case when $\log^k d \ge n$, but in the general case it is too expensive for our purposes. In the next section, we will show a randomised implementation of $\prefixsearch$ which requires both less space and less time. 

\begin{remark}
We will use the $k$-errata tree to retrieve the patterns that are within Hamming distance $k$ from the query string $Q$ or from one of its prefixes. Recall that we mark each node of the $k$-errata tree corresponding to an end of a dictionary pattern. Furthermore, during the preprocessing step, we compute a pointer from each node to its nearest marked ancestor. At the end of each $\prefixsearch$ we follow the pointers and retrieve the patterns corresponding to the marked nodes between the end and the start of the $\prefixsearch$. The number of the $\prefixsearch$ operations that we perform does not change.
\end{remark}

\subsection{Randomised implementation of the k-errata tree}\label{sec:randomisation}
We are now ready to define the randomised $k$-errata tree. Apart from the sketches for the Hamming distance, we will use the following result of Belazzougui et al.~\cite{zfasttrie}:

\begin{theorem}[$z$-fast tries~\cite{zfasttrie}]\label{th:zfast}
Consider a string $S$ and suppose that we can compute the Karp–-Rabin fingerprint of any prefix of $S$ in $t_\varphi$ time. A compact trie on a set of $r$ strings of length at most $n$ can be stored in $\Oh(r)$ space to support the following queries in $\Oh(t_\varphi \cdot \log n)$ time: Given $S$, find the highest node $v$ such that the longest prefix of $S$ present in the trie is a prefix of the label of the root-to-$v$ path. The answer is correct with high probability. 
\end{theorem}

In other words, the $z$-fast trie can be considered as a randomised implementation of the compact trie. In the randomised implementation of the $k$-errata tree, we replace each compact trie with the $z$-fast trie. This gives an efficient implementation of all $\prefixsearch$ queries if $u$ is the root of a compact trie, however, the general case requires more work.

\begin{lemma}\label{lm:randomised-k-errata}
Assume $\log^k d < n$. A dictionary of $d$ patterns of maximal length $m$ can be preprocessed into a data structure which we call \emph{randomised $k$-errata tree} that uses $\Ooh(k d \log^k d)$ space and allows retrieving all the patterns that are within Hamming distance $k$ from $Q$ or one of its prefixes in $\Ooh(k \log^k d + |\occ|)$ time, assuming that we know the $k$-mismatch sketches of all prefixes of~$Q$. The error is two-sided and the answer is correct with high probability.
\end{lemma}
\begin{proof}
Recall from above that the $k$-errata tree is a collection of compact tries. In the randomised version of the $k$-errata tree, we replace each of them with a $z$-fast trie. We also store the $k$-mismatch sketch of the label of every node of the tries, which requires $\Ooh(k d \log^k d)$ space in total. 

We now describe the query algorithm. Recall that each dictionary look-up with~$k$ mismatches is a sequence of calls to the $\prefixsearch$ procedure, and therefore it suffices to give an efficient implementation of $\prefixsearch$. We first explain how to implement this operation if it starts at the root of some compact trie of the $k$-errata tree. Since we know the $k$-mismatch sketches of the prefixes of $Q$, we know their Karp--Rabin fingerprints. Hence, by Fact~\ref{fct:KR_concat} we can retrieve the Karp--Rabin fingerprint of any substring of~$Q$ in $\Oh(1)$ time. Theorem~\ref{th:zfast} immediately implies that a $\prefixsearch$ starting at the root of a compact trie can be implemented in $\Oh(\log n)$ time. Note that if the end of the $\prefixsearch$ is an implicit node, then the functionality of the $z$-fast tries will allow us retrieving only the edge this node belongs to, but not the node itself. As we show below, it is sufficient for our purposes. 

We now give an implementation of a $\prefixsearch$ starting at an arbitrary node of a compact trie by reducing it first to a $\prefixsearch$ that starts at an explicit node of the trie and then to a $\prefixsearch$ that starts at the root of the trie. We first show a reduction from a $\prefixsearch$ that starts at an implicit node to a $\prefixsearch$ that starts at an explicit node. As explained above, we might know the edge this starting node belongs to, but not the node itself. However, from the description of the query algorithm in Section~\ref{sec:reminderkerratatree} it follows that the algorithm will continue along the edge by running $\prefixsearch$ operations until it either runs out of the mismatch credit or reaches the explicit node at the lower end of the edge. We will fast-forward to the lower end of the edge using the $k$-mismatch sketches. Namely, let $Q'$ be the query string when we entered the current tree (note that we do not change the tree when retrieving patterns of group~\rom{3}). Importantly, the string $Q'$ is a suffix of $Q$. We want to check whether we can reach the lower end of the edge and not run out of the mismatch credit. In other words, we want to compare the number of mismatches between the label $S$ of the lower end of the edge and the prefix $S'$ of $Q'$ of length $|S|$, and the mismatch credit. We use the $k$-mismatch sketches for this task. We store the sketch of $S$, and the sketch of $S'$ can be computed in $\Ooh(k)$ time as it is a substring of~$Q$. Having computed the sketches, we can compute the Hamming distance between $S$ and $S'$ using Lemma~\ref{lm:compute_dist}. If the Hamming distance is larger than the available mismatch credit, we stop, otherwise, we continue the $\prefixsearch$ from the explicit node at the lower end of the edge. 

Finally, we show an implementation of a $\prefixsearch$ for a string $Q'$ that starts at an explicit node $u$ of a trie. Let $S$ be the label of $u$. Our task is equivalent to performing a $\prefixsearch$ starting from the root of a trie for a string $S \circ Q'$. Recall that Theorem~\ref{th:zfast} assumes that we can extract the Karp--Rabin fingerprint of any prefix of $S \circ Q'$. We do not know the Karp--Rabin fingerprints of the prefixes of $S \circ Q'$, but we can compute them as follows. First, we use the $k$-mismatch sketches similar to above to compute the at most $k$ mismatches that occurred on the way from the root of the trie to~$u$. After having computed the mismatches, we can compute any of the fingerprints in $\Ooh(k)$ time by taking the fingerprint of the corresponding substring of $Q$ and ``fixing'' it in at most~$k$ positions. By fixing we mean that if two strings $X,Y$ differ in positions $i_1, i_2, \ldots, i_{k'}$, where $k' \le k$, then $\varphi(X)$ equals $\varphi(Y) + \sum_{j = 1}^{k'} (X[i_j]-Y[i_j]) r^{i_j-1} \bmod p$, where $r,p$ are as in Definition~\ref{def:kr}, and therefore can be computed in $\Oh(k)$ time. The value $\varphi(X^R)$ can be computed analogously from $\varphi(Y^R)$, which implies that the Karp--Rabin fingerprint of $X$ can be computed from that of $Y$ in $\Oh(k)$ time.

The bounds on the space occupied by the data structure and the query time follow from Lemmas~\ref{lm:kerratatreesize} and~\ref{lm:kerratatreetime}. The bound on the error probability follows from the assumption $\log^k d < n$,  Lemma~\ref{lm:compute_dist}, and Theorem~\ref{th:zfast}.
\end{proof}

\section{Algorithm}
In this section, we show our streaming algorithm for dictionary matching with $k$ mismatches. Similar to previous work on streaming pattern matching, we assume that we receive the patterns first, preprocess them (without accounting for the preprocessing time), and then receive the text.

\begin{theorem}\label{th:main-deamortised}
Assume an alphabet of size $n^{\Oh(1)}$. There is a randomised streaming algorithm that solves dictionary matching with $k$ mismatches in $\Ooh(k d \log^k d)$ space and $\Ooh(k \log^k d + |\occ|)$ time per arriving letter. The algorithm has two-sided error and its answers are correct with high probability.
\end{theorem}

Hereafter we assume $k \log \log d < \log^k d$ (all logs are in base two), which is true for any~$d \ge 3$ and $k \ge 1$. For $d = 1,2$ we can use Corollary~\ref{cor:streaming-k-mismatch} to achieve the complexities of Theorem~\ref{th:main-deamortised}. We also assume that $\log^k d < n$, as otherwise we can use the deterministic $k$-errata tree (Section~\ref{sec:reminderkerratatree}): namely, during the preprocessing step the algorithm builds the $k$-errata tree for the reverses of the patterns. During the main stage, at every position of the text the algorithm runs a dictionary look-up with $k$ mismatches for the suffix of the text of length $n$. The algorithm uses $\Oh(nd + d \log^k d) = \Oh(d \log^k d)$ space and $\Oh(n + \log^k d \log \log n + |\occ|) = \Ooh(k \log^k d + |\occ|)$ time per letter.

\subsection{Outline of the algorithm}
The main idea is to consider periodic and non-periodic patterns separately. 

\begin{definition}[$k$-period, Clifford et al.~\cite{k-mismatch}]
The $k$-period of a string $S = S[1] \ldots S[m]$ is the minimal integer $\pi > 0$ such that the Hamming distance between $S[\pi+1,m]$ and $S[1,m-\pi]$ is at most $2k$.
\end{definition}

\begin{observation}\label{obs:rare}
If the $k$-period of $S$ is larger than $d$, there can be at most one $k$-mismatch occurrence of $S$ per $d$ consecutive positions of the text.
\end{observation}

For patterns of length smaller or equal to $3d$, we use the algorithm based on the randomised $k$-errata tree (Lemma~\ref{lm:streaming-k-errata}). The remaining patterns form two smaller dictionaries: the first dictionary $\mathcal{D}_1$ contains the patterns $P_i$ such that the $k$-period of their suffix $\tau_i = P_i[|P_i|-2d+1,|P_i|]$ is larger than $d$, and the second dictionary $\mathcal{D}_2$ contains patterns $P_i$ such that the $k$-period of their suffix $\tau_i$ is at most~$d$. The algorithm processes the dictionaries in parallel. To process $\mathcal{D}_1$, it makes use of Observation~\ref{obs:rare}. To process $\mathcal{D}_2$, it uses the fact that the patterns and therefore the regions of the text containing their $k$-mismatches occurrences are periodic and can be encoded in small space.

The idea of exploiting periodicity has flavour similar to~\cite{streamdictionary,GP17,k-mismatch,Clifford:17}, but we make a significant step forward to allow \emph{both} mismatches and multiple patterns. 

The rest of the section is organised as follows. First, we show two algorithms that we will use as auxiliary routines (Section~\ref{sec:tools}). Second, we show an algorithm for the dictionary $\mathcal{D}_1$ (Section~\ref{sec:large}) and an algorithm for the dictionary $\mathcal{D}_2$ (Section~\ref{sec:small}). These two algorithms run in parallel give a streaming algorithm with the complexities of Theorem~\ref{th:main-deamortised}, however, the time complexity is amortised. In Section~\ref{sec:deamortise} we show how to de-amortise the algorithm which yields Theorem~\ref{th:main-deamortised}.

\subsection{Tools}\label{sec:tools}
We now describe two algorithms. The first algorithm is a streaming algorithm that is based on the randomised $k$-errata tree. The space complexity of this algorithm depends on the maximal length of the patterns linearly, so we cannot use it in the general case, but it will be handy when processing short patterns. The second algorithm is a streaming algorithm that can check whether the current stream ends with a $k$-mismatch occurrence of one of the dictionary patterns on demand. 

\subsubsection{Tool: Algorithm based on the randomised $k$-errata tree}\label{sec:algkerratatree}
\begin{lemma}\label{lm:streaming-k-errata}
Assume $\log^k d < n$. Assuming that all patterns are of length at most $\ell \le n$, there is a streaming algorithm for dictionary matching with $k$ mismatches that uses $\Ooh(k \cdot (\ell+d \log^k d))$ space and $\Ooh(k \log^k d + |\occ|)$ time per letter. The algorithm is randomised with two-sided error, and its answers are correct with high probability.
\end{lemma}
\begin{proof}
During the preprocessing step, the algorithm builds the $k$-errata tree for the reverses of the patterns. During the main step, the algorithm maintains the $k$-mismatch sketches of the reverses of the $\ell$ longest prefixes of the text in the round-robin fashion updating them in $\Ooh(k)$ time when a new letter arrives (Corollary~\ref{cor:sketch_concat}). If the text ends with a $k$-mismatch occurrence of some pattern $P_i$, there is a suffix of the text of length $|P_i| \le \ell$ such that the Hamming distance between it and some pattern in the dictionary is bounded by $k$. It means that we can retrieve all occurrences of such patterns by using the randomised $k$-errata tree for the reverse of the $\ell$-length suffix of the text. We can retrieve the $k$-mismatch sketch of the reverse of any substring of this suffix in $\Ooh(k)$ time (Corollary~\ref{cor:sketch_concat}), and therefore perform the dictionary look-up query in $\Ooh(k \log^k d + |\occ|)$ time. In total, the algorithm uses $\Ooh(k \cdot (\ell + d \log^{k} d))$ space and $\Ooh(k \log^k d + |\occ|)$ time per letter.  \end{proof}

\subsubsection{Tool: On-demand algorithm}\label{sec:ondemand} 
Clifford et al.~\cite{k-mismatch} showed that the problem of detecting $k$-mismatch occurrences of a pattern $P$ in the text can be reduced to dictionary matching in the following way. Let $R$ be a set of $\log n$ primes chosen u.a.r. from the interval $[k \log^2 n, 102 k \log^2 n]$. A subpattern $P_{r}^\ell$ of the pattern $P$ is defined by a prime $r \in R$ and an integer $1 \le \ell \le r$, namely, $P_{r}^\ell = P[\ell] P[r+\ell] P[2r+\ell] \dots$ and so on until the end of $P$. (If $P$ is shorter than $r$, some subpatterns are left undefined.) Furthermore, let $Q$ be the set of all the primes in the interval $[\log n, 3 \log n]$, and consider a second-level partitioning of the patterns. Namely, a subpattern $P_{r,q}^\ell$ for $r \in R, q \in Q$ and $1 \le \ell \le q \cdot r$ is defined as $P_{r,q}^\ell = P[\ell] P[r\cdot q+\ell] P[2r\cdot q+\ell] \dots$ and so until the end of $P$. (Again, some subpatterns can be undefined.)

\begin{corollary}[of Lemmas 3.2 and 5.2~\cite{k-mismatch}]\label{lm:reduction_dm}
Consider an alignment of a pattern $P$ of length at most $n$ and the text. Given the subset of the subpatterns of $P$ that match exactly at this alignment, there is a randomised $\Ooh(k)$-time algorithm that outputs a message ``The distance is larger than~$k$.'' if the Hamming distance between $P$ and the text is larger than $k$, and the true value of the Hamming distance otherwise. The algorithm is correct with probability at least $1-2/n^2$.
\end{corollary}
\begin{proof}
First, for every $r$, the algorithm finds the number $\mu_r$ of subpatterns $P_{r}^\ell$ that do not match. As $\max_{r \in R} \mu_r$ is a lower bound for the Hamming distance between $P$ and the text, in the case $\max_{r \in R} \mu_r > k$  the algorithm can simply output the message ``The distance is larger than~$k$.''
Otherwise, if $\max_{r \in R} \mu_r \le k$, the algorithm computes the Hamming distance as described in~\cite[Lemma 3.2]{k-mismatch}. This requires $\Ooh(k)$ time and gives the correct value of the Hamming distance with probability at least $1-1/n^2$ conditioned on the fact that the distance is bounded by $2k$. By~\cite[Lemma 5.2]{k-mismatch} this condition is violated with probability at most $1/4n^2$. Therefore, the answer is correct with probability at least $1-2/n^2$.  
\end{proof}

We also exploit the following result:

\begin{theorem}[{\cite[Theorem 2]{GP17}}]\label{th:streaming_dm}
Given a dictionary of $d$ patterns $P_1, P_2, \ldots, P_d$ of lengths at most $n$ and a text $T$ of length $\Oh(n)$ over an alphabet of size $n^{\Oh(1)}$. There exists a randomised streaming algorithm for the dictionary
matching problem that uses $\Oh(d \log n)$ space and $\Ooh(1)$ time per letter of the text. At each time moment, the algorithm outputs the longest occurrence of a dictionary pattern that is a suffix of the current text. The algorithm has one-sided error and its answers are correct with high probability.
\end{theorem}

\begin{corollary}\label{cor:ondemand}
Given a dictionary of $d$ patterns $P_1, P_2, \ldots, P_d$ of lengths at most $n$ and a text $T$ of length $\Oh(n)$. 
There is a randomised streaming algorithm that uses $\Ooh(k^2 d)$ space and processes each letter of the text in~$\Ooh(1)$ time. On demand, the algorithm can tell in~$\Ooh(k)$ time if there is a $k$-mismatch occurrence of a pattern $P_i$ that ends at the current position of text. The algorithm has two-sided error and its answers are correct with high probability. 
\end{corollary}
\begin{proof}
During the preprocessing step, we build two compact tries. The first trie contains the reverses of all subpatterns $(P_i)^\ell_{r}$, where $r \in R, 1 \le \ell \le r$, and the second one the reverses of all subpatterns $(P_i)^\ell_{r,q}$, where $r \in R, q \in Q, 1 \le \ell \le q \cdot r$. Furthermore, we traverse the tries in the depth-first order and memorise, for each node, the first and the last time when we see it. This allows us to tell in $\Oh(1)$ time if the reverse of one subpattern is a prefix of the reverse of another subpattern. (Recall that the preprocessing time is not accounted for in the time complexity.)

During the main stage, we partition the text stream $T$ into substreams. For each $r \in R$ and $1 \le \ell \le r$, we define the text substream $T_r^\ell = T[\ell]T[ r+\ell]T[2 r+\ell]\dots$ and so on until the end of $T$. On the substream $T_{r}^\ell$ we run the algorithm of Theorem~\ref{th:streaming_dm} for the dictionary of subpatterns $(P_i)^{\ell'}_{r}$, where $i = \{1, 2, \ldots, d\}$ and $1 \le \ell' \le r$.
Also, for each pair of primes $q \in Q, r \in R$ and an integer $1 \le \ell \le q \cdot r$ we define a text substream $T_{q, r}^\ell = T[\ell]T[q \cdot r+\ell]T[2q \cdot r+\ell]\dots$ and so on until the end of $T$. On the substream $T_{q,r}^\ell$ we run the algorithm of Theorem~\ref{th:streaming_dm} for the dictionary of subpatterns $(P_i)^{\ell'}_{q,r}$, where $i = \{1, 2, \ldots, d\}$ and $1 \le \ell' \le q \cdot r$.

In total there are $\Ooh(k)$ substreams and $\Ooh(k d)$ subpatterns per substream, and therefore the algorithm uses $\Ooh(k^2 d)$ space. To process each letter of the text, the algorithm requires $\Ooh(1)$ time: when a new letter $T[p]$ arrives, the algorithm must update one substream $T_r^\ell$ for each $r \in R$ ($\ell = p \bmod r$) and one substream $T_{r,q}^\ell$ for each pair $r \in R, q \in Q$ ($\ell = p \bmod (r \cdot q)$). 

Using the output of the dictionary matching algorithms and the compact tries built at the preprocessing step, we can check, for any subpattern, if it matches at the current alignment in $\Oh(1)$ time and therefore can decide if there is a $k$-mismatch occurrence of $P_i$ in $\Ooh(k)$ time by Lemma~\ref{lm:reduction_dm}.
\end{proof}

\subsection{Streaming algorithm for patterns with large periods}\label{sec:large}
In this section, we show a streaming algorithm for the dictionary $\mathcal{D}_1$ that contains patterns~$P_i$ such that the $k$-period of their suffix $\tau_i = P_i[|P_i|-2d+1,|P_i|]$ is at least $d$. 

\begin{lemma}\label{lm:largeperiods}
Assume $\log^k d < n$. There is a randomised streaming algorithm that retrieves all $k$-mismatch occurrences of the patterns from $\mathcal{D}_1$ in the text. The algorithm uses $\Ooh(k d \log^k d)$ space and $\Ooh(k \log^{k} d)$ amortised time per letter. The algorithm has two-sided error and its answers are correct with high  probability.
\end{lemma}
\begin{proof}
Note that any $k$-mismatch occurrence of a pattern $P_i$ ends with a $k$-mismatch occurrence of $\tau_i$. We retrieve the occurrences of $\tau_i$ via the algorithm based on the $k$-errata tree (Lemma~\ref{lm:streaming-k-errata}). At each position of the text, the algorithm outputs all indices~$i$ such that there is a $k$-mismatch occurrence of $\tau_i$ ending at this position. If the same index is output less than $d$ positions apart, i.e. there are two $k$-mismatch occurrences of $\tau_i$ that end at positions $p_1$ and $p_2$ such that $|p_1-p_2| \le d$, there is an error and we stop the execution of the algorithm. When we find a $k$-mismatch occurrence of $\tau_i$, our second step is to check if it can be extended into a $k$-mismatch occurrence of $P_i$ which we do with the help of the on-demand algorithm (Corollary~\ref{cor:ondemand}).  

We now analyse the algorithm. Assume that neither the $k$-errata tree nor the on-demand algorithm err. The $k$-errata tree for the suffixes~$\tau_i$ occupies $\Ooh(k d \log^k d)$ space. To report $d' \le d$ $k$-mismatch occurrences of $\tau_i$'s that end at the current position of the text, we  spend $\Ooh(k \log^k d + d')$ time  (Lemma~\ref{lm:streaming-k-errata}). The on-demand algorithm uses $\Ooh(k^2 d)$ space and $\Ooh(1)$ time per letter (Corollary~\ref{cor:ondemand}). To test if a $k$-mismatch occurrence of $\tau_i$ can be extended into a $k$-mismatch occurrence of $P_i$, we need $\Ooh(k)$ time. Since by Observation~\ref{obs:rare}, unless the $k$-errata tree algorithm errs, there is at most one $k$-mismatch occurrence of~$\tau_i$ per $d$ positions of the text, the time bound follows.

Lemma~\ref{lm:streaming-k-errata}, Corollary~\ref{cor:ondemand}, and the union bound imply that the answers of the algorithm are correct with high probability (assuming that the error probabilities in Lemma~\ref{lm:streaming-k-errata} and Corollary~\ref{cor:ondemand} are chosen to be small enough).
\end{proof}

\subsection{Streaming algorithm for patterns with small periods}\label{sec:small}
In this section, we show a streaming algorithm for the second dictionary $\mathcal{D}_2$ that contains patterns $P_i$ such that the $k$-period of  their suffix $\tau_i = P_i[|P_i|-2d+1,|P_i|]$ is at most $d$.

\begin{lemma}\label{lm:smallperiods}
Assume $\log^k d < n$. There is a randomised streaming algorithm that retrieves all $k$-mismatch occurrences of the patterns from $\mathcal{D}_2$ in the text. The algorithm uses $\Ooh(k d \log^k d)$ space and $\Ooh(k \log^{k} d + |\occ|)$ amortised time per letter. The algorithm has two-sided error and its answers are correct with high probability.
\end{lemma}

We define $\tau'_i$, $|\tau_i'| \ge |\tau_i|$, to be the longest suffix of $P_i$ with the $k$-period at most $d$. Two cases are possible: 

\begin{enumerate}
\item\label{it:i} The suffix $\tau'_i$ equals $P_i$ (in other words, the $k$-period of $P_i$ is at most $d$).
\item\label{it:ii} The suffix $\tau'_i$ does not equal $P_i$. 
\end{enumerate}

\subsubsection{Algorithm for Case~\ref{it:i}}
We first assume that Case~\ref{it:i} holds for all the patterns, and then  extend the algorithm to Case~\ref{it:ii} as well. We start by showing a simple but important property of patterns with small periods.

\begin{figure}
\begin{center}
\begin{tikzpicture}[scale=0.5]
\node at (-0.5,0) {\ldots};
\draw (0,0) rectangle (18,1);
\foreach \x in {1,2,3,4} {
	\draw (\x*4,0)--(\x*4,1);
}
\draw[dashed] (5,0)--(5,1);
\draw (5,1)--(5,2);
\draw (16,1)--(16,2);
\draw[<->] (5,1.75)--(16,1.75) node[pos=0.5, above] {$n$ positions};

\draw[blue] (10,-0.1) rectangle (15.9,1.1) node[pos=0.5] {$L$};

\node[below] at (17.8,0) {$r$};
\node[below] at (5,0) {$j \cdot d -n+1$};
\node[below] at (16,0) {$j \cdot d$};

\draw [decorate,decoration={brace,amplitude=10pt,mirror}]
(10,-1) -- (18,-1) node [black,midway,yshift=-0.6cm] {$L \circ T[j\cdot d+1, r]$};
\end{tikzpicture}
\end{center}
\caption{If there is a $k$-mismatch occurrence of a pattern $P_i$ that ends at a position $r$ of $T$, and the $k$-period of $P_i$ is at most $d$, then the occurrence is contained in the suffix $L \circ T[j\cdot d+1, r]$ of the text.}
\label{fig:periodic_suffix}
\end{figure}

\begin{lemma}\label{lm:region-periods}
Consider a position $r$ of the text $T$. Let $j \cdot d$ be the largest multiple of~$d$ that is smaller than $r$ and $L$ be the longest suffix of $T[j\cdot d-n+1, j \cdot d]$ with the $2k$-period at most $d$. Every $k$-mismatch occurrence of $P_i \in \mathcal{D}_2$ in $T$ that ends at the position $r$ is fully contained in $L \circ T[j\cdot d+1, r]$ (see Fig.~\ref{fig:periodic_suffix}). 
\end{lemma}
\begin{proof}
Consider an occurrence $T[\ell,r]$ of a pattern $P_i \in \mathcal{D}_2$ that ends at the position~$r$. Since the length of $P_i$ is at most $n$, $\ell \ge r-n+1 > j\cdot  d-n+1$. Now, let $\rho \le d$ be the $k$-period of $P_i$ and $\Ham(X,Y)$ the Hamming distance between strings $X,Y$. We have $\Ham(T[\ell+\rho-1,r], P_i[\rho,|P_i|]) \le k$ and $\Ham(T[\ell,r-\rho+1], P_i[1,|P_i|-\rho+1]) \le k$. Therefore, by the triangle inequality,  
$$\Ham(T[\ell+\rho-1,r], T[\ell,r-\rho+1]) \le 2k + \Ham(P_i[\rho,|P_i|], P_i[1,|P_i|-\rho+1])$$
As $\Ham(P_i[\rho,|P_i|], P_i[1,|P_i|-\rho+1]) \le 2k$, we obtain that $\Ham(T[\ell+\rho-1,r], T[\ell,r-\rho+1]) \le 4k$ and hence the $2k$-period of $T[\ell,j\cdot d]$ is at most $\rho \le d$. Consequently, it is contained in~$L$. The claim follows.
\end{proof}

In other words, it suffices to know $L$ and the at most $d$ last positions of the text to be able to retrieve the $k$-mismatch occurrences of the patterns that end at a position $r$. Of course, in general $L$ can be long and we cannot store it explicitly. However, as it is close to periodic, we will be able to encode it in small space.

\begin{lemma}\label{lm:encoding}
Consider a string $L$ such that its  $2k$-period $\rho$ is at most $d$, and recall that~$L^R$ stands for the reverse of $L$. There is an $\Ooh(k^2 d)$-space encoding of $L$ that allows to retrieve the $4k$-mismatch sketch of any substring of $L^R$, given by its endpoints, in $\Ooh(k)$ time.
\end{lemma} 
\begin{proof}
Consider a partitioning of $L^R$ into non-overlapping blocks of length $\rho$ (the last block may be shorter). We say that a block contains a mismatch if, for some $i$, its $i$-th letter is different from the $i$-th letter of the preceding block. For convenience, we also say that the first and the last blocks of $L^R$ are mismatch-containing. Note that the number of blocks containing a mismatch is $\Oh(k)$ (this is because $\Ham(L^R[1,|L|-\rho+1],L^R[\rho+1,|L|]) = \Ham(L[1,|L|-\rho+1],L[\rho+1,|L|]) \le 4k$, and it upper bounds the number of the blocks containing a mismatch).

We encode $L^R$ as a sorted array of the starting positions of all blocks containing a mismatch. For each mismatch-containing block, we store the $4k$-mismatch sketch of each of its suffixes, as well as the sketch of the suffix of $L^R$ that starts right before the block. In total, the encoding occupies $\Ooh(k^2 d)$ space. 

By Corollary~\ref{cor:sketch_concat}, it suffices to show that we can compute the  $4k$-mismatch sketch of any suffix of $L^R$ in $\Ooh(k)$ time. We retrieve the sketch in the following way. Let $\ell$ be the starting position of the suffix, and note that each mismatch-containing block starts a streak of equal blocks of $L^R$. First, we find the streak of blocks $\ell$ belongs to, and retrieve the sketch of the suffix of $L^R$ starting just after the streak in $\Ooh(k)$ time. The remaining part consists of a number of repetitions of the block containing the position $\ell$ prepended with the suffix of the block (see Fig.~\ref{fig:encoding}). We can compute the sketch of the block and of its suffix in $\Ooh(k)$ time, and therefore we can compute the sketch of the remaining part in $\Ooh(k)$ time using Corollary~\ref{cor:sketch_concat}. 
\end{proof}

\begin{figure}
\begin{center}
\begin{tikzpicture}[scale=0.25]
\draw (0,0) rectangle (34,1);

\draw[red,pattern=north west lines, pattern color=red] (0.2,0.2) rectangle (3.8,0.8);
\draw[red,pattern=north west lines, pattern color=red] (4.2,0.2) rectangle (7.8,0.8);
\draw[green,pattern=north west lines, pattern color=green] (8.2,0.2) rectangle (11.8,0.8);
\draw[green,pattern=north west lines, pattern color=green] (12.2,0.2) rectangle (15.8,0.8);
\draw[green,pattern=north west lines, pattern color=green] (16.2,0.2) rectangle (19.8,0.8);
\draw[blue,pattern=north west lines, pattern color=blue] (20.2,0.2) rectangle (23.8,0.8);
\draw[blue,pattern=north west lines, pattern color=blue] (24.2,0.2) rectangle (27.8,0.8);
\draw[blue,pattern=north west lines, pattern color=blue] (28.2,0.2) rectangle (31.8,0.8);
\draw[black,pattern=north west lines, pattern color=black] (32.2,0.2) rectangle (33.8,0.8);

\foreach \x in {1,2,3,4,5,6,7,8} {
	\draw (\x*4,0)--(\x*4,1);
}

\node[below] at (10,0) {$\ell$};
\draw[dashed] (10,0)--(10,1);
\draw (10,1)--(10,2);
\draw (34,1)--(34,2);
\draw[<->] (10,1.75)--(34,1.75) node[pos=0.5, above] {a suffix of $L^R$};

\draw [decorate,decoration={brace,amplitude=10pt,mirror}]
(20,-1) -- (34,-1) node [black,midway,yshift=-0.6cm] {sketch of this suffix is known};
\end{tikzpicture}
\end{center}
\caption{The string $L^R$ and retrieval of the $4k$-mismatch sketch of a suffix of $L^R$. There are four streaks of equal blocks shown in different colours, each streak starts with a mismatch-containing block.}
\label{fig:encoding}
\end{figure}

We are now ready to explain the algorithm. 
During the preprocessing stage, we build the randomised $k$-errata tree for the reverses of all the patterns. During the main stage of the algorithm, we maintain the suffix $L$ and its encoding.
We initialize $L$ with an empty string and update it each $d$ letters. While reading the next $d$ letters of the text, that is a substring $T[(j-1) \cdot d+1, j \cdot d]$, we compute the $4k$-mismatch sketches of the reverses of its $d$ prefixes in $\Ooh(kd)$ time (Corollary~\ref{cor:sketch_concat}). After having reached $T[j \cdot d]$, we update~$L$, the longest suffix of $T[j\cdot d-m+1, j\cdot d]$ with the $2k$-period at most $d$, and its encoding:

\begin{lemma}
$L$ and its encoding can be updated in $\Ooh(k^2)$ amortised time per letter. 
\end{lemma}
\begin{proof}
$L$ is determined by its endpoints in $T$. To update $L$, it suffices to find the longest suffix of $T[j \cdot d-n+1, j \cdot d]$ such that the Hamming distance between it and its copy shifted by $\rho$ positions, for $\rho = 1,\dots,d$, is at most $2k$. For a fixed value of $\rho$, we use binary search and the $4k$-mismatch sketches. 

Suppose we want to decide whether the Hamming distance between $T[\ell, j \cdot d-\rho+1]$ and $T[\ell+\rho, j \cdot d]$ is at most $4k$. Note that if $T[\ell, j \cdot d]$ is longer than $L \circ T[(j-1) \cdot d+1,j \cdot d]$, then its $2k$-period is larger than $d$. This is because if the $2k$-period of $T[\ell, j \cdot d]$ is at most $d$, then the $2k$-period of $T[\ell, (j-1) \cdot d]$ is at most~$d$, and hence $T[\ell, (j-1) \cdot d]$ must be shorter than $L$. In other words, we must only consider the case when $T[\ell, j \cdot d]$ is fully contained in $L \circ T[(j-1)\cdot d+1, j\cdot d]$. 

In this case, both $T[\ell, j\cdot d-\rho+1]$ and $T[\ell+\rho, j\cdot d]$ can be represented as a concatenation of a suffix of $L$ and a substring of $T[(j-1) \cdot d+1, j \cdot d]$. We can retrieve the $4k$-mismatch of the reverse of any suffix of $L$ and the $4k$-mismatch sketch of the reverse of any substring of $T[(j-1)\cdot d+1, j \cdot d]$ in $\Ooh(k)$ time using  Lemma~\ref{lm:encoding} and Corollary~\ref{cor:sketch_concat}. Therefore, we can compute the $4k$-mismatch sketches of  $(T[\ell, j\cdot d-\rho+1])^R$ and $(T[\ell+\rho, j\cdot d])^R$ and hence $\Ham((T[\ell, j\cdot d-\rho+1])^R, (T[\ell+\rho, j\cdot d])^R) = \Ham(T[\ell, j\cdot d-\rho+1],T[\ell+\rho, j\cdot d])$ in $\Ooh(k)$ time using Lemma~\ref{lm:compute_dist}. In total, we need $\Ooh(d k)$ time to update~$L$, or $\Ooh(k)$ amortised time per letter. 

We can now update the encoding of $L$ in $\Ooh(k^2 d)$ time: Using the $4k$-mismatch sketches for $L^R[\rho,|L|]$ and $L^R[1,|L|-\rho+1]$, we can find the $\Oh(k)$ blocks containing a mismatch in $\Ooh(k)$ time. We can then re-build the sorted array of the starting positions of mismatch-containing blocks in $\Ooh(k)$ time and compute the sketches for them in $\Ooh(k^2 d)$ time,  or~$\Ooh(k^2)$ amortised time per letter.
\end{proof}

Let $T[r]$ be the latest arrived letter of the text. To retrieve the $k$-mismatch occurrences that end at the position $r$, we use the $k$-errata tree for the reverses of the patterns in~$\mathcal{D}_2$ that we built during the preprocessing stage. Let $j \cdot d$ be the largest multiple of~$d$ that is at most $r$ and let $L$ be defined as above. By Lemma~\ref{lm:region-periods}, any $k$-mismatch occurrence of pattern $P_i \in \mathcal{D}_2$ that ends at~$r$ must be either equal to a suffix of $T[j\cdot d+1,r]$, or the concatenation of some suffix of $L$ and $T[j \cdot d+1,r]$. The encoding of $L$ allows to compute the $4k$-mismatch sketch (and therefore the $k$-mismatch sketch) of the reverse of any suffix of $L$ in $\Ooh(k)$ time. We can also compute the $4k$-mismatch sketch of the reverse of any of the $d$ latest suffixes of the text in $\Ooh(k)$ time. Therefore, we can retrieve the $k$-mismatch occurrences of the patterns for a current position in $\Ooh(k \log^{k} d + |\occ|)$ time using the $k$-errata tree. In total, the algorithm for Case~\ref{it:i} uses $\Ooh(k d \log^k d)$ space and $\Ooh(k \log^{k} d+ |\occ|)$ amortised time per letter. 

\subsubsection{Extension to Case~\ref{it:ii} and wrapping up}
Consider now Case~\ref{it:ii}. Note first that the $2k$-period of a string $P_i[|P_i|-|\tau'_i|, |P_i|]$, which is $\tau'_i$ extended by one letter, must be at least $d$, and therefore by Observation~\ref{obs:rare} there can be at most one $k$-mismatch occurrence of $P_i[|P_i|-|\tau'_i|, |P_i|]$ per $d$ positions of the text. We use the techniques of the algorithm for Case~\ref{it:i} to retrieve the $k$-mismatch occurrences of $P_i[|P_i|-|\tau'_i|, |P_i|]$, and then use the on-demand algorithm (Corollary~\ref{cor:ondemand}) to check which of the retrieved occurrences can be extended into full occurrences of the patterns $P_i$.

In more detail, consider a position $r$ of the text. As before, let $j \cdot d$ be the largest multiple of $d$ that is smaller than $r$ and $L$ be the longest suffix of $T[j\cdot d-m+1, j\cdot d]$ with the $2k$-period at most $d$. Let now $L'$ be the suffix $L$ extended by one letter to the left, i.e. $L' = T[j \cdot d - |L|, j \cdot d]$. By definition, the $(2k+1)$-period of $L'$ is at most $d$. Furthermore, similar to Lemma~\ref{lm:region-periods}, we can show that any $k$-mismatch occurrence of $P_i[|P_i|-|\pi'_i|,|P_i|]$ ending at the position $r$ must be fully contained in $L' \circ T[j \cdot d+1,r]$. 

Therefore, instead of the suffix $L$, we can maintain the encoding of the longest suffix $L''$ such that its $(2k+1)$-period is at most $d$ via an algorithm similar to that of the previous section in $\Ooh(k^2 d)$ space and $\Ooh(k^2)$ amortised time per letter. (Note that $L''$ contains $L'$.) Using the encoding, we can retrieve the $4k$-mismatch (and therefore $k$-mismatch) sketch of the reverse of any substring of $L'' \circ T[j\cdot d+1,r]$ in $\Ooh(k)$ time and hence we can find the $k$-mismatch occurrences of $P_i[|P_i|-|\tau'_i|, |P_i|]$ using the $k$-errata tree in $\Ooh(k \log^{k} d + |\occ|)$ time per letter. If the $k$-errata tree reports two occurrences of $P_i[|P_i|-|\tau'_i|, |P_i|]$ at distance less than $d$ from each other, there is an error and we stop. Otherwise, we check whether the reported occurrence can be extended into a $k$-mismatch occurrence of $P_i$ via the on-demand algorithm in $\Ooh(k)$ time.

In total, the algorithm for Case~\ref{it:ii} uses $\Ooh(k d \log^k d + k^2 d) = \Ooh(k d \log^k d)$ space and $\Ooh(k \log^{k} d+ |\occ|)$ amortised time per letter. Lemma~\ref{lm:smallperiods} follows.

\subsection{De-amortisation}\label{sec:deamortise}
Lemma~\ref{lm:largeperiods} and Lemma~\ref{lm:smallperiods} yield a streaming algorithm for the dictionary matching with~$k$ mismatches with the complexities of Theorem~\ref{th:main-deamortised}, except that the time is amortised. Below we explain how to de-amortise the algorithm. We use a standard approach called the \emph{tail trick}. 

\subsubsection{De-amortised algorithm with a delay}
First, note that there is an easy way to de-amortise the algorithm of Lemma~\ref{lm:largeperiods} if we allow a delay by $d$ letters.
Formally speaking, this means that an occurrence ending at position $i$ in the text does not need to be reported immediately after
reading the $i$-th letter, but instead can be reported after reading any of the subsequent $d$ letters (and we do not require any
control on when precisely does this happen).
To obtain such an algorithm, we divide the text into non-overlapping blocks of length $d$, and de-amortise the processing time of a block over the next block. We must memorize the $k$-mismatch occurrences of the suffixes $\tau_i$ that end at the last $2d$ positions of the text, but this requires only $\Oh(d)$ space and we can afford it. 

We now show how to de-amortise the algorithm for Case~\ref{it:i} of Lemma~\ref{lm:smallperiods}. This time, we do not need the delay. The only step of the algorithm that requires de-amortisation is updating $L$ and its encoding. We can de-amortise this step in a standard way. Namely, we de-amortise the time we need for an update over the next $d$ letters of text. We also maintain the sketches of the reverses of the $2d$ longest prefixes of the text in the round-robin fashion using $\Ooh(k d)$ space and $\Oh(k)$ time. If we need to extract the sketch of the reverse of some suffix of $L$ before the update is finished, we use the previous version of the data structure and the sketches of the reverses of the prefixes to compute the required values using Corollary~\ref{cor:sketch_concat}.  

Finally, we show how to de-amortise the algorithm of Case~\ref{it:ii} of Lemma~\ref{lm:smallperiods}, again with a delay of $d$ letters. Recall that this algorithm first finds the $k$-mismatch occurrences of the suffixes $P_i[|P_i|-|\tau'_i|, |P_i|]$ using an algorithm similar to the algorithm for Case~\ref{it:i} of Lemma~\ref{lm:smallperiods}, which can be de-amortised with no delay as explained above, and then tests these occurrences using the on-demand algorithm (Corollary~\ref{cor:ondemand}), which can be de-amortised with a delay of $d$ letters. Importantly, there are at most $d$ occurrences that need to be tested per $d$ letters, so we can memorize them until we can test them. The claim follows. 

\subsubsection{Removing the delay}
We now show how to remove the delay. Recall that we assume the patterns to have lengths larger than $3d$. We partition each pattern $P_i = H_i \circ Q_i$, where $Q_i$ is the suffix of~$P_i$ of length~$d$, and~$H_i$ is the remaining prefix. The idea is to find occurrences of the prefixes $H_i$ and of the suffixes $Q_i$ independently, and then to see which of them form an occurrence of $P_i$.

As above, we have three possible cases: the $k$-period of $H_i[|H_i|-2d+1, |H_i|]$ is larger than~$d$; the $k$-period of $H_i$ is at most $d$; the $k$-period of $H_i$ is larger than $d$ but the $k$-period of $H_i[|H_i|-2d+1, |H_i|]$ is at most $d$. 

In the second case, we do not need to change much. For the current position $r$ of the text we consider the largest $j \cdot d$ such that $r - j \cdot d \ge d$ and define $L$ to be the longest suffix of $T[j \cdot d - n+1, j \cdot d]$ such that its $2k$-period is at most $d$. We store the $k$-errata tree on the reverses of $P_i = H_i \circ Q_i$ and run the de-amortised algorithm described in the previous section that maintains the suffix $L$. Any $k$-mismatch occurrence of a pattern~$P_i$ is fully contained in $L \circ T[j\cdot d+1, r]$, and therefore we can find all such occurrences using the $k$-errata tree as above. 

We now explain how we remove the delay in the first and third cases. 
To find the occurrences of $Q_i$ we use the algorithm based on the $k$-errata tree (Lemma~\ref{lm:streaming-k-errata}). To find the occurrences of $H_i$, we use the de-amortised version of the algorithm of Lemma~\ref{lm:largeperiods} or of Lemma~\ref{lm:smallperiods}, as appropriately, that report the occurrences with a delay of at most $d$ letters. It means that at the time when we find an occurrence of $Q_i$, the corresponding occurrence of $H_i$ is already reported, so it is easy to check whether they form an occurrence of $P_i$. The only technicality is that we need to store the occurrences of $H_i$ that we found while processing the last $d$ letters of the text. 

To this end, we use a dynamic hashing scheme~\cite{Dietzfelbinger1992}. The scheme allows to store a dynamic dictionary in linear space and with high probability guarantees constant look-up and update times. The answers to the look-up queries are always correct. Note that we can modify the data structure slightly to have worst-case constant time per operation if we allow the answers to be correct only with high probability (which we can afford), namely, if an operation takes too much time, we can simply abandon~it.

We use the scheme for each of the last $d$ positions of the text. Namely, consider a position~$p$ of the text and suppose that we found a set of $k$-mismatch occurrences of the prefixes $H_i$ that end at $p$. Consider one of the prefixes, $H_i$, and let the Hamming distance between a prefix $H_i$ and the text be $h \le k$. Recall that by Lemma~\ref{lm:kerratatreetime} there are $\Oh(\log^k d)$ nodes of the $k$-errata tree labelled by $Q_i$. For each such node $u$ of the $k$-errata tree, we insert a pair $(u, h)$ into the dictionary. In case we insert a pair $(u, h)$ several times for different prefixes $H_i$'s, we associate $(u, h)$ with the set of such prefixes. Note that at any moment the total size of the dictionaries is $\Ooh(d \log^k d)$ as each of the patterns $H_i$ has at most one $k$-mismatch occurrence over each~$d$ consecutive positions of the text.

Suppose we are at a position $p$ of the text and we have run a dictionary look-up query and found the $\Oh(\log^k d)$ nodes in the tries of the $k$-errata tree corresponding to the suffixes $Q_i$ that occur at this position with at most $k$ mismatches. For each such node $u$ we know the Hamming distance $h'$ between the occurrences and the text. We then go to the dictionary at the position $(p-d)$ and look up pairs $(u,k-h'), (u,k-h'-1), \dots, (u,0)$. If they are in the dictionary, we report all $H_i$'s associated with these pairs. This step takes $\Oh(k \log^{k} d + |\occ|)$ time.

\section{Space lower bound}\label{sec:space_lower}
We now show a space lower bound that demonstrates that for constant $k$ our algorithm is optimal up to a polylogarithmic factor:

\begin{lemma}\label{lm:space_lower_bound}
Any streaming algorithm for dictionary matching with $k$ mismatches such that its answers are correct with constant probability requires $\Omega(k d)$ bits of space.
\end{lemma}
\begin{proof}
In the communication complexity setting the Index problem is stated as follows. We assume that there are two players, Alice and Bob. Alice holds a binary string of length $n$, and Bob holds an index $i$. In a one-round protocol, Alice sends Bob a single message (depending on her input and on her random coin flips) and Bob must compute 
 the $i$-th bit of Alice's input using her message and his random coin flips correctly with probability at least $2/3$. The length of Alice's message (in bits) is called the randomised one-way communication complexity of the problem. The randomised one-way communication complexity of the Index problem is $\Omega(n)$~\cite{Kremer:1995:ROC:225058.225277}.

Given a streaming algorithm for dictionary matching with $k$ mismatches, one can construct a randomised one-way communication complexity protocol for the Index problem as follows. As above, let $d$ be the size of the dictionary, and assume that~$n = kd$. Split Alice's string into~$d$ blocks of length $k$. Let $\#, \$, \$_1, \ldots, \$_d$ be distinct letters different from $\{0,1\}$.  For the $j$-th block $B_j$ create a string $P_j = (\$_j)^{k+1} \# B_j$. For Bob's input $i = k \cdot q + r$ we create a string $T$ which is equal to $(\$_q)^{k+1}$ concatenated with a string of length $k+1$ obtained from $\$^{k+1}$ by changing the $(r+1)$-th bit to $0$. A streaming dictionary matching with $k$ mismatches for the set of patterns $P_i$ and $T$  will output a $k$-mismatch occurrence of $B_q$ at the position $2k+2$ of the text iff the $r$-th bit of Alice's input is equal to $0$. Therefore, if Alice preprocesses $P_j$ as in the streaming algorithm and sends the result to Bob, Bob will be able to continue to run the streaming algorithm on $T$ to decide the $i$-th bit of Alice's input. Therefore, the lower bound for communication complexity of the Index problem is a space lower bound for streaming dictionary matching with mismatches. Lemma~\ref{lm:space_lower_bound} follows. 
\end{proof}

\bibliographystyle{plain}      
\bibliography{main-streaming}   
\end{document}